\documentclass[12pt]{article}

\usepackage[OT1]{fontenc}
\usepackage{amsthm,amsmath}

\usepackage{graphicx}
\usepackage{url}
\usepackage{mathrsfs}      
\usepackage{mathptmx}      
\usepackage{subfigure}
\usepackage{latexsym,amsmath}
\usepackage{euscript}
\usepackage{amsfonts}
\usepackage{moreverb}
\usepackage{xcolor}
\usepackage[latin1]{inputenc}
\usepackage[OT2,T1]{fontenc}

\usepackage{algorithm}
\usepackage[noend]{algorithmic}
\algsetup{indent=2em}

\usepackage{amsfonts}
\usepackage{moreverb}
\usepackage{url}
\usepackage{pgf}

\newtheorem{proposition}{Proposition}

\renewcommand{\c}[1]{\mathcal{#1}}
\newcommand{\bb}[1]{\mathbb{#1}}
\renewcommand{\v}[1]{\boldsymbol{#1}}
\newcommand{\m}[1]{\mathbf{#1}}
\newcommand{\dnet}{d_\mathrm{net}}
\newcommand{\Var}{\mathrm{Var}}
\newcommand{\g}{\leftarrow}
\newcommand{\RE}{\mathrm{RE}}
\newcommand{\WNRV}{\mathrm{WNRV}}
\begin{document}

\begin{center}
{\Large\bf
 Reliability Estimation for Networks with \\[5pt] Minimal Flow Demand and Random Link Capacities\\[15pt]}

{\bf Zdravko I. Botev}, The University of New South Wales, Sydney, Australia \\
{\bf Pierre L'Ecuyer}, Universit\'e de Montr\'eal, Canada, and Inria--Rennes, France \\
{\bf Bruno Tuffin}, INRIA Rennes---Bretagne Atlantique, France

\end{center}

{\bf Abstract: }
We consider a network whose links have random capacities and in which a certain target amount
of flow must be carried from some source nodes to some destination nodes.
Each destination node has a fixed demand that must be satisfied and each source node has a given supply.
We want to estimate the unreliability of the network, defined as the probability that 
the network cannot carry the required amount of flow
to meet the demand at all destination nodes.
When this unreliability is very small, which is our main interest in this paper, 
standard Monte Carlo estimators become useless because failure to meet the demand is a rare event.
We propose and compare two different methods to handle this situation,
one based on a conditional Monte Carlo approach and the other based on generalized splitting.
{We find that the first is more effective when the network is highly reliable and  
not too large, whereas for a larger network and/or moderate reliability, the second is more effective.}

{\bf Keywords:}
{network reliability},
{stochastic flow network},
{Conditional Monte Carlo},
{permutation Monte Carlo},
{generalized splitting}


\section{Introduction}

Network reliability estimation problems are commonplace in various application areas 
such as transportation, communication, and power distribution systems; see for example \cite{pGER10a}.
In many of those problems, the states of certain network components are subject to uncertainty 
and there is a set of conditions under which the network is operational,
and one wishes to estimate the network \emph{unreliability}, defined as the probability $u$
that the network is in a failed state (i.e., is not operational).
When $u$ is very small, a standard (crude) Monte Carlo (MC) approach that merely generates 
the component states, computes the indicator function that the network is operational or not,
and averages over $n$ independent runs to estimate $u$, is unsatisfactory because 
the relative error (defined as the standard deviation of the estimator divided by the expected value $u$)
of the MC estimator goes to infinity when $u\to 0$.

One reliability problem that has received a lot of attention is the \emph{static network reliability} 
estimation problem, in which each link of the network is failed with a given probability 
and the network is operational when a given (specific) subset of the nodes are all connected.
Effective estimation methods have been developed for this problem when $u$ is small;
see \cite{vBOT13a,vBOT16m,vELP91a,vFIS86a,pGER10a,vLOM99a} and the references therein.
The model considered in this paper is more general.
Instead of having only a binary state (up or down), each link has a random capacity
that can take many possible values, 
there is a fixed demand that must be satisfied at certain nodes (called the destination nodes),
a fixed supply is available at some other nodes (the source nodes),
and the network is operational when it can carry the flow to satisfy all the demands.
As a special case, there can be a single source node and a single destination node, with a fixed demand,
and the network is operational when the maximum flow that can be sent from the source to the destination 
reaches the demand. We will describe our methods in this particular setting to simplify the notation,
but the methods apply to the general setting as well.
The case of links with binary states is a special case.
The several methods developed for this special case do not readily apply to the 
network flow setting considered here, but we show how two of the best available methods for the binary case,
\emph{permutation Monte Carlo} (PMC) and \emph{generalized splitting} (GS),
can be adapted to this problem.  
The adaptation is not straightforward.

The PMC method \cite{vELP91a,pGER10a,vLOM99a} constructs an artificial continuous-time Markov chain (CTMC) defined as follows.
Each capacity is assumed to have a discrete distribution over a finite set of possible values.
This can approximate a continuous distribution if needed.
We assume that all links start at their minimal capacity, and the capacity of one link may increase
each time the CTMC has a jump.  The CTMC is constructed so that the probability that the network 
is failed at time 1 is equal to $u$.
PMC generates the discrete-time Markov chain (DTMC) underlying the CTMC,
i.e., only the sequence of states that are visited until the network is operational,
and conditional on that sequence it computes the probability that the network is failed at time 1,
as an estimator of $u$.
This conditional probability can be computed by exploiting the property that 
the failure time has a phase-type conditional distribution,
whose cumulative distribution function (cdf) and density can be expressed in terms of matrix exponentials.
We show how to adapt and apply the PMC principle to our problem.
The CTMC construction is quite different than for the binary case.
We also prove, under certain conditions, that the resulting PMC estimator has 
bounded relative error (BRE) when $u \to 0$ for a given network.

GS \cite{vBOT10a} is a rare-event estimation method where the rare event is the intersection of a nested sequence of
events and its probability is the product of conditional probabilities. Each conditional probability is estimated thanks to resampling strategies, making the overall estimation  more accurate than a direct estimation of the rare-event probability itself.  The application of GS to this problem was discussed in \cite{vBOT14a} for the situation
in which the capacities have a continuous distribution, and experimental 
results were reported for a small example.
But the GS algorithm proposed there does not work in general when the capacities
have a discrete distribution.
We show however that GS can be applied in the discrete case if we combine it with the same
CTMC construction as for PMC.
The GS algorithm does not have BRE in the asymptotic regime when $u\to 0$,
but it becomes more efficient that PMC when the size of the network increases.
The relative error typically increases (empirically) as $\c O(-\log u)$.

The remainder of the paper is organized as follows. 
In Section~\ref{sec:model}, we formulate the network flow model considered in this paper.
In Section~\ref{sec:pmc}, we construct a CTMC which permits one to apply PMC to this model,
for the case where each capacity is distributed over a finite set.
In Section~\ref{sec:gs}, we explain how to apply GS to this model.
We report numerical experiments in Section~\ref{sec:examples}.
{Our experimental results agree with the fact that PMC has BRE when $u\to 0$,
under appropriate conditions.  It can accurately estimate extremely small values of $u$
when the network is not too large.
When the network gets larger and $u$ is not too small, on the other hand, 
GS becomes more effective than PMC.}

\section{The model}
\label{sec:model}

Let $\c G=(\c V,\c E)$ be a graph with a set of nodes $\c V$ and a set of links $\c E$ with cardinality $m=|\c E|$.
For $i=1,\dots,m$, link $i$ has a random integer-valued flow capacity $X_i$ with discrete marginal 
distribution
$
  p_i(x) = \bb P[X_i=x] 
$
over the set
\[
  \c X_i = \{c_{i,0},\ldots,c_{i,b_i}\}, \qquad  0\le c_{i,0} < c_{i,1} < \cdots < c_{i,b_i} < \infty.
\]
This is a standard assumption; see \cite{vBOT14a} and the references therein.
Thus, the random network state $\v X=(X_1,\ldots,X_m)$ belongs to the space $\c X=\prod_{i=1}^m\c X_i$ 
and has joint pdf $p(\v x)=\bb P(\v X=\v x)$, for $\v x\in\c X$. 
We also make the standard independence assumption (see \cite{vALE92a,vBUL02a,sCHO14a})  
that $p(\v x)=\prod_{i=1}^m \bb P[X_i=x_i]$ and that the nodes do not fail.

To keep the notation and the exposition simple, in the remainder of the paper
we describe the model and the methods under the assumption that there is a single source
and a single destination.  
The generalization to multiple sources and destinations is straightforward, as explained below.
The fixed demand level at the destination is $d_\mathrm{net}>0$ 
and the maximum flow that can be carried from the source to the destination is a random variable $\Psi(\v X)$,
which is a function of the link capacities.
The well-known max-flow min-cut theorem says that the maximum value of a flow from a source to a destination is equal to the minimum capacity of a cut  in the network. 
Efficient algorithms are available to compute $\Psi(\v X)$; for example the Ford-Fulkerson algorithm.

We are interested in estimating the unreliability of the flow network, defined here as 
\[
  u = \bb P[\Psi(\v X) < d_\mathrm{net}] = \sum_{\{\v x\in\c X : \Psi(\v x) < d_\mathrm{net}\}} p(\v x);
\]
that is, the probability that the maximum flow $\Psi(\v X)$ fails to meet the demand. 
This problem was considered in \cite{sFIS96a}, for example. 
In the particular case where $\c X_i = \{0,1\}$ for each $i$ and $d_\mathrm{net} = 1$,
we have an instance of the static network reliability problem mentioned in the introduction,
with the source and destination as the selected set of nodes to be connected.

To generalize to multiple sources and destinations, we would assume a fixed demand $d_i$ at
each destination node $i$, a fixed supply $s_i$ at each source node $i$, 
and the event $\{\Psi(\v X) < d_\mathrm{net}\}$ would be replaced by the event that 
the network does not have sufficient capacity to send flow to satisfy all the demands from the available supplies. 

For small networks, it is possible to compute and store most of  
the  minimal cutsets or pathsets and use them to obtain exact or approximate values for $u$; 
see \cite{pJAN10a,pZUO07a} for example.  
But for large networks, no polynomial-time algorithm is known for computing $u$ exactly \cite{pCOL87a},  
and one must rely on approximations or on estimation via Monte Carlo.
Of particular interest is the situation in which the network is highly reliable, 
i.e., $u$ is a very small rare-event probability, 
because crude Monte Carlo then becomes ineffective.   

Several Monte Carlo variance-reduction methods have been proposed for network reliability estimation 
in rare-event situations;
see, e.g., \cite{vBOT13a,vBOT16m,vCAN03a,vELP91a,pGER10a,vRAM05a,vLEC11c,vTUF14a} 
and the references given there.
Most of these methods are for the special case of independent links with binary states 
and nodes that never fail.
Some have been extended to links with three possible states 
\cite{pGER14c,pGER14b,pGER14a}, but this remains restrictive.
We now describe how two of the most efficient methods, PMC and GS, can be adapted to our model (an archived version of this report: \cite{botev:hal-01745187}).

\section{Reformulating the model as a CTMC and applying PMC}
\label{sec:pmc}

We now show how to construct an artificial CTMC for this static model,
which will permit us to apply PMC as described in the introduction.
This CTMC construction differs from that used in \cite{vBOT16m,pGER10a}.

\subsection{Constructing the CTMC}
\label{sec:pmc-ctmc}

For each $i$, let $\bb P(X_i=c_{i,k}) = p_i(c_{i,k}) = r_{i,k} > 0$ for $k=0,\dots,b_i$.
Define independent exponential random variables $Y_{i,1},\ldots,Y_{i,b_i}$
with rates $\lambda_{i,1},\ldots,\lambda_{i,b_i}$, respectively, where the $\lambda_{i,k}$
still have to be chosen.
Suppose that the capacity of link $i$ is 
$c_{i,0}$ from time $T_{i,0}=0$ to time $T_{i,1} = \min(Y_{i,1}, \ldots, Y_{i,b_i})$ (exclusive), 
after that it is $c_{i,1}$ from time $T_{i,1}$ to time $T_{i,2} = \min(Y_{i,2}, \ldots, Y_{i,b_i})$, 
it is $c_{i,2}$ from time $T_{i,2}$ to time $T_{i,3} = \min(Y_{i,3}, \ldots, Y_{i,b_i})$, 
and so on, and finally it is $c_{i,b_i}$ from time $T_{i,b_i}$ to $T_{i,b_i+1} = \infty$.
Under this process, the capacity of link $i$ at time $\gamma\ge 0$ is given by 
\begin{eqnarray}
  X_i(\gamma) &=& c_{i,k}  \quad \mbox{ for } T_{i,k} \le\gamma < T_{i,k+1} \mbox{ and } 0\le k\le b_i \\
     &=& \max_{k}\{c_{i,k} : T_{i,k} \le\gamma\}.
 \label{eq:Xi-capacity}
\end{eqnarray}
The times $T_{i,1}, \ldots, T_{i,b_i}$ are not necessarily all distinct; often, many of them are equal,
so that the number of jumps at which the capacity changes can be much smaller than $b_i$.
For example, if $T_{i,1} = Y_{i,b_i}$, then we have $T_{i,1} = T_{i,2} = \cdots = T_{i,b_i}$.
As another example, if $b_i=3$ and $Y_{i,2} < Y_{i,1} < Y_{i,3}$, then 
$0 < T_{i,1} = T_{i,2} < T_{i,3}$ and the capacity of link $i$ jumps 
from $c_{i,0}$ to $c_{i,2}$ at time $T_{i,2} = Y_{i,2}$ and jumps again
from $c_{i,2}$ to $c_{i,3}$ at time $T_{i,3} = Y_{i,3}$.
In general, the process $\{X_i(\gamma),\, \gamma\ge 0\}$ has an upward jump 
at each of the \emph{distinct} jump times $T_{i,k}$.

To show that this process is a CTMC, suppose that we are at time $\gamma\ge 0$ and $X_i(\gamma) = c_{i,k}$.
Then we know that $Y_{i,k} \le\gamma$ and that $Y_{i,\ell} > \gamma$ for all $\ell > k$.  
The $Y_{i,\ell}$ for $\ell < k$ can be anything, but they have no influence on the process trajectory
after time $\gamma$. This means that the current state $X_i(\gamma)$ contains all the relevant information
that needs to be known at time $\gamma$ to generate the future of the process.

The capacity $X_i(\gamma)$ of link $i$ at time $\gamma\ge 0$ satisfies
\[
  \bb P[X_i(\gamma) \le c_{i,k}] = \bb P[\min(Y_{i,k+1}, \ldots,  Y_{i,b_i}) > \gamma] 
    = \exp[-\gamma (\lambda_{i,k+1} + \dots + \lambda_{i,b_i})].
\]
If we select the $\lambda_{i,k}$'s so that the last expression equals 
$r_{i,0} + \cdots + r_{i,k}$ for each $k$ when $\gamma=1$, 
then $X_i(1)$ has the exact same distribution as $X_i$,
the capacity of link $i$ in the original static model.
This is equivalent to having
\[
  \lambda_{i,k+1} + \dots + \lambda_{i,b_i} = - \ln (r_{i,0} + \cdots + r_{i,k}).
\]
To achieve this, it suffices to put 
\begin{eqnarray}
 \lambda_{i,b_{i}}    &=& - \ln (r_{i,0} + \cdots + r_{i,b_i-1}) ~=~ - \ln(1-r_{i,b_i})
 \label{eq:lambdaibi}
\end{eqnarray}
and then for $k = b_i-1,\ldots, 1$ (in descending succession): 
\begin{eqnarray}
  \lambda_{i,k} 
   &=& - \ln (r_{i,0} + \cdots + r_{i,k-1}) - \lambda_{i,k+1} - \dots - \lambda_{i,b_i}.
 \label{eq:lambdaik}
\end{eqnarray}
Note that (\ref{eq:lambdaik}) can be rewritten as
$\lambda_{i,k} = - \ln (r_{i,0} + \cdots + r_{i,k-1}) + \ln (r_{i,0} + \cdots + r_{i,k})$,
which can never be negative.  We have proved the following.
\begin{proposition}
If we select $\lambda_{i,b_i}, \lambda_{i,b_i-1}, \ldots, \lambda_{i,1}$ according to 
(\ref{eq:lambdaibi}) and (\ref{eq:lambdaik}) and the process $X_i(\cdot)$ as in (\ref{eq:Xi-capacity}), 
then $\lambda_{i,k} \ge 0$ for each $k$, $\{X_i(\gamma),\, \gamma\ge 0\}$ is a CTMC process, and 
$X_i(1)$ has exactly the same discrete distribution as the capacity of link $i$ in the original model:
$\bb P[X_i(1) = c_{i,k}] = r_{i,k}$ for $k = 0,\dots,b_i$.
As a result, $\v X(1) = (X_1(1),\dots,X_m(1))$ has the same distribution as $\v X$
and one has $u = \bb P[\Psi(\v X(1)) < d_{\rm net}]$.
\end{proposition}

\subsection{Applying PMC}
\label{sec:network-flow-pmc}

Under the assumption that all links are independent, a simple way of applying PMC to this model is as follows.
Generate all the $Y_{i,k}$'s independently with their rates $\lambda_{i,k}$, 
put them in a large vector 
\[
 \v Y = (Y_{1,1},\dots,Y_{1,b_1}, \ldots, Y_{m,1},\dots,Y_{m,b_m}) 
\]
of size $\kappa = b_1 + \cdots + b_m$, and sort this vector in increasing order to obtain 
\[
  Y_{\pi(1)}\le Y_{\pi(2)} \le \dots \le Y_{\pi(\kappa)},
\]
where $\pi(j) = (i,k)$ if $Y_{i,k}$ is in position $j$ in the sorted vector,
so that $\pi = (\pi(1),\dots,\pi(\kappa))$ can be seen as the permutation of the pairs $(i,k)$ 
that corresponds to the sort.  This permutation gives an ordering of the $\kappa$ pairs $(i,k)$.
When scanning those pairs in the given order, each pair $(i,k)$ corresponds to a potential capacity increase for link $i$.
The capacity increases if and only if no pair $(i,k')$ for $k' > k$ has occurred before.
Conditional on $\pi$, one can add those pairs in the given order and update the capacities accordingly,
until the maximum flow in the network reaches $d_{\mathrm{net}}$.
Suppose this occurs when adding the pair $(i,k) = \pi(C)$ for some integer $C > 0$.
Let $T_C = Y_{\pi(C)}$.
The (unbiased) conditional (PMC) estimator of $u$ is then 
\[
  \bb P[\Psi(\v X(1)) < \dnet \mid \pi]
	= \bb P[T_C > 1 \mid \pi] 
  = \bb P[T_C > 1 \mid \pi(1),\dots,\pi(C)] 
  = \bb P[A_1 + \cdots + A_{C} > 1],
\]
where $A_1 = Y_{\pi(1)}$ is an exponential random variable with rate 
$\Lambda_1 = \sum_{i=1}^m \sum_{k=1}^{b_i} \lambda_{i,k}$,
each $A_j = Y_{\pi(j)} - Y_{\pi(j-1)}$ is an exponential random variable with rate 
$\Lambda_j = \Lambda_{j-1} - \lambda_{\pi(j-1)}$ for $j = 2,\dots,C$, and these $A_j$'s are independent.
Given $\pi$ and $C$, $T_C = A_1 + \cdots + A_{C}$ is the sum of $C$ independent exponential 
random variables with rates $\Lambda_1,\dots,\Lambda_C$, which has a phase-type distribution,
whose complementary cdf is given by
\begin{equation}
\label{eq:sum-cdf}
 1-F(\gamma \mid\pi) = \bb P[T_C > \gamma \mid \pi] = \v e_1^\top \exp(\m Q \,\gamma)\v 1
\end{equation}
where $\v e_1^\top = (1,0,\dots,0)$, $\v 1 = (1,\dots,1)^\top$ (the $^\top$ means ``transposed''), and
\[
  \m Q = \begin{pmatrix}
    -\Lambda_1  & \Lambda_1   &    0      & \cdots         &  0      &  \\
         0      &  -\Lambda_2 & \Lambda_2 &   0 	         &  \vdots &  \\
         0      &      0      &   \ddots  &  \ddots        &  0      &  \\
      \vdots    &      0      &      0    & -\Lambda_{C-1} &  \Lambda_{C-1} \\
         0      &  \cdots     &      0    &    0           &  -\Lambda_C \\
	\end{pmatrix}.
\]
Reliable and fast computation of (\ref{eq:sum-cdf}) is discussed in \cite{vBOT13a,vBOT16m}.
%

To compute the critical number $C$ at which the flow reaches the demand,
we must be able to update efficiently the maximum flow in the network each time
we increase the capacity of one link. 
We do this as explained in Section 4 of \cite{vBOT14a}. We refer to this algorithm as the \emph{incremental maximum flow algorithm}.

To estimate $u$ by PMC, for a fixed threshold $\dnet$, we simulate
$n$ independent realizations $W_1,\dots,W_n$ of 
\begin{equation}
\label{eq:cdf-cmc}
  W = W(\pi) = \bb P[T_C > 1 \mid \pi] 
\end{equation}
and take the average $\bar W_n = (1/n)\sum_{i=1}^n W_i$.
Compared with the crude Monte Carlo estimator that would take the indicator
$I = \bb I[\Psi(\v X(1)) < \dnet]$ in place of $W$, it is always true that 
$\Var[W] < \Var[I]$, because $W = \bb E[I \mid \pi]$.

The estimators discussed so far are for a single (fixed) demand $\dnet$.
With PMC, it is also possible to estimate $u = u(\dnet)$ as a function of the demand $\dnet$,
over some interval, using the same simulations for all demands.
To do this, for any given permutation $\pi$, we can compute $C = C(\dnet)$ as a function of the 
demand over the interval of interest.  This would be a step function, often with just a few jumps.
Then we compute $W$ for all values of $C$ that are visited over this interval.
This provides an estimator $W(\dnet)$  of $u(\dnet)$ as a function of $\dnet$. 
By averaging the $n$ realizations of this estimator, we obtain a functional estimator of $u(\dnet)$ 
over the interval of interest.

\subsection{Improved PMC}
\label{sec:pmc-improved}

The PMC strategy described earlier can be improved by removing some useless jumps.
First, whenever $c_{i,k} \ge \dnet$ for $k < b_i$, we can immediately remove all the jumps 
$(i,k+1),\dots,(i,b_i)$, because when the capacity of a link has reached $\dnet$, 
it is useless to increase it further.
Capacity levels larger than $\dnet$ can in fact be all reset to $\dnet$ right away in the model, 
the probability of $\dnet$ in the new model being taken as the  the probability of values larger 
than $\dnet$ in the initial model.
For simplicity, when the demand is fixed, we assume in our algorithm that this has been done already,
so that $c_{i,b_i} \le \dnet$ for all $i$, and then there is no need to remove those useless capacity levels.

Second, the jump times $Y_{i,k}$ that do not change the capacity of link $i$ can also be removed.  
That is, whenever $\pi(j) = (i,k)$ and the capacity of link $i$ has already 
reached a value $c_{i,k'} > c_{i,k}$, i.e., $(i,k') = \pi(j')$ for some $j' < j$, then there is no 
need to consider the pair $(i,k)$ when it is encountered in the permutation, so we can remove the corresponding jump.

Let $\tilde\pi$ be the permutation obtained after removing all those pairs $(i,k)$ from the sorted vector,
and $\tilde C$ the corresponding value of $C$ in this reduced permutation.
As soon as the max flow reaches $d_{\mathrm{net}}$, we have found $\tilde C$.
When we encounter $\tilde\pi(j) = (i,k)$ and the previous capacity of link $i$ was $c_{i,k'} < c_{i,k}$,
the capacity of link $i$ jumps to $c_{i,k}$ and 
we must decrease $\Lambda_j$ by $\tilde\lambda_{i,k} = \lambda_{i,k'+1} + \cdots + \lambda_{i,k}$,
because the jumps that correspond to $(i,k'+1),\dots,(i,k)$ can now be removed from consideration.

\begin{algorithm}
\caption{:  PMC algorithm for multi-state flow network
  \label{algo:pmc-anti}}
\begin{algorithmic}[1]
 \FOR {$i=1$ \TO $m$}
   \STATE{draw $Y_{i,1},\dots,Y_{i,b_i}$ with the appropriate rates $\lambda_{i,k}$}
	 \STATE{let $\c L_i = \{(i,1),\ldots,(i,b_i)\}$}
	 \STATE{${\min} \g (i,b_i)$ and $\tilde\lambda_{i,b_i} \g \lambda_{i,b_i}$}
   \FOR {$k=b_i-1$ \TO $1$}
	   \IF{$Y_{i,k} > Y_{\min}$}
	     \STATE{remove $(i,k)$ from the list $\c L_i$}
		   \STATE{$\tilde\lambda_{\min} \g \tilde\lambda_{\min} + \lambda_{i,k}$}
       \STATE{$S_{i,k} \g 0$}    \COMMENT{ this jump is deactivated}
	   \ELSE
	     \STATE{${\min} \g (i,k)$}   \COMMENT{ this pair $(i,k)$ is retained}
		   \STATE{$\tilde\lambda_{\min} \g \lambda_{i,k}$}
       \STATE{$S_{i,k} \g 1$}    \COMMENT{ this jump is activated}
		\ENDIF
   \ENDFOR
 \ENDFOR
 \STATE{merge the sorted lists $\c L_1,\dots,\c L_m$ into a single list sorted by increasing order, 
        $Y_{\tilde\pi(1)},\ldots,Y_{\tilde\pi(\tilde\kappa)}$}
 \STATE{$\Lambda_1 \g \lambda_{1,1} + \cdots + \lambda_{1,b_1} + \cdots + \lambda_{m,b_m}$}
 \STATE{$j \g 0$}
 \STATE{$\v X \g (c_{1,0},\dots, c_{m,0})$}
 \WHILE{maximum flow $\Psi(\v X) < d_{\mathrm{net}}$}
   \STATE{$j \g j+1$}
   	\IF{$S_{\tilde\pi(j)} = 1$}  
   		\STATE{$(i,k) \g \tilde\pi(j)$}  \COMMENT{this jump has not been removed or executed}
		  \STATE{$\Lambda_{j+1} \g \Lambda_{j} - \tilde\lambda_{i,k}$}
		  \STATE{$S_{i,k} \g 0$}
			\STATE{$X_i \g c_{i,k}$}    \COMMENT{increase capacity of $i$-th link}
		  \STATE{Filter()}    \COMMENT{do nothing (default), or FilterSingle, FilterAll, etc.}
	  \ENDIF
 \ENDWHILE
 \STATE{$\tilde C \g j$}   \COMMENT{the critical jump number}
 \RETURN {$W \g \bb P\left[A_1 + \cdots + A_{k-1} > 1 \mid \tilde\pi, \tilde C\right]$.}
\end{algorithmic}
\end{algorithm}

Algorithm~\ref{algo:pmc-anti} describes this reduced version of PMC in a more formal way.
It returns one realization of $W$.  
Indentation delimits the scope of the loops.
In the first {\bf for} loop, for each link $i$, the algorithm generates the exponential
random variables $Y_{i,k}$ and then immediately eliminates those that correspond to (useless) jump times 
at which the capacity of the link does not change. 
This preliminary filtering is very easy and efficient to apply and may eliminate a significant
fraction of the jumps, especially for links that have many capacity levels. 
The remaining jumps are sorted in a single list (for all links) and each one receives
a Boolean tag $S_{\tilde\pi(j)}$, initialized to 1, which means that this jump is currently scheduled to occur. 

Then these jumps are ``executed'' in chronological order,
by increasing the corresponding capacities, until the critical jump number $\tilde C$ is found.
After that, $W$ can be computed.
The Boolean variables $S_{\tilde\pi(j)}$ are used in the optional Filter() subroutine,
which can be used to try to eliminate further useless jumps 
after a jump is executed and the corresponding capacity is increased (this is discussed in 
Section~\ref{sec:turnip}).
Algorithm~\ref{algo:pmc-anti} would be invoked $n$ times, independently, 
and $u$ would be estimated by the average of the $n$ realizations of $W$.

Other variants of the algorithm can be considered
and some might be more efficient, but this is not completely clear.
For example, instead of generating all the variables $Y_{i,k}$ at the beginning,
one may think of generating the permutation $\pi$ directly without generating those $Y_{i,k}$,
as was done in \cite{pGER10a} for the binary case.
This appears complicated and we did not implement it.

\subsection{Removing jumps having no impact on maximum flow}
\label{sec:turnip}

In Algorithm~\ref{algo:pmc-anti}, in the case where Filter() does nothing, 
all pairs $(i,k)$ for which the capacity of link $i$ increases are retained in $\tilde\pi$
and the corresponding jumps are executed.
But it sometimes occurs that increasing the capacity of link $i$ to $c_{i,k}$ (or more) is useless 
because it can no longer have an impact on the event that the maximum flow exceeds the demand or not. 
In this case, one can cancel (deactivate) all the future jumps related to the capacity of link $i$.
In our implementation, these future jumps are canceled by setting their Boolean variables $S_{i,k}$ to 0.
Increasing the capacity of link $i$ is useless in particular
if it is already possible to send $d_{\mathrm{net}}$ units of flow between
the two nodes connected by link $i$.
This obviously happens if the capacity of link $i$ is already of $d_{\mathrm{net}}$, 
which is trivial to verify, but under our assumption 
(made at the beginning of Section~\ref{sec:pmc-improved})
that a link has no capacity level above $\dnet$, 
this cannot happen, and our algorithm ignores this possibility.

Increasing the capacity of link $i$ is also useless when $d_{\mathrm{net}}$ units of flow can be sent in total,
either directly on link $i$ or indirectly via other links.  
This is generally harder (more costly) to verify.
To detect it, one can run a max-flow algorithm to compute how much flow can be sent between these two nodes.
This can be done each time the capacity of a link is increased.

Algorithm~\ref{algo:filterJumps1} does this only for the link $i$ whose capacity has just been increased, 
at each step $j$.  Since the link $i$ generally changes at every step $j$, we have a different max-flow 
problem (for a different pair of nodes) at each step.  For this reason, in our implementation we recompute the 
max-flow from scratch at each step $j$.  Of course, this brings significant overhead.

Algorithm~\ref{algo:filterJumps2} is even more ambitious: it computes the max-flow between nodes 
for all pairs of nodes.  Then for each link $i$ for which the current max-flow between the corresponding
two nodes meets the demand, it cancels all future jumps associated with that link.
Doing this at each step $j$ might be too costly, so in our implementation the user selects a positive 
integer $\nu$ and does it only at every $\nu$ steps, i.e., when $j$ is a multiple of $\nu$.
We compute the max-flow for all pairs of nodes using the algorithm of \cite{oGUS90a}, 
which is a simplified variant of the Gomory-Hu method \cite{oGOM61a}.
This algorithm computes the all-pairs max-flow by applying $\left|\mathscr{V}\right| - 1$ times 
a max-flow algorithm for one pair of nodes, which is generally more efficient than applying 
a max-flow algorithm for each link.
Algorithm~\ref{algo:filterJumps2} also recomputes the maximum flows from scratch each time it is called,
rather than reusing computations from the previous time and just updating the max flows.
(In fact, the $\left|\mathscr{V}\right| - 1$ pairs of nodes for which the max-flow is computed in the algorithm
change from one call to the next, and we are not aware of an effective incremental algorithm that would
reuse and just update the previous computations.)



\begin{algorithm}
\caption{: FilterSingle
  \label{algo:filterJumps1}}
\begin{algorithmic}
\STATE{$f \g $ compute maximum flow between terminal nodes of link $i$}
\IF{$f \geq d_{\mathrm{net}}$}
	\FOR{$k = 1$ \TO $b_i$}   
		\IF{$S_{i, k} = 1$}
			\STATE{$S_{i, k} \g 0$}
			\STATE{$\Lambda_{j+1} \g \Lambda_{j+1} - \tilde\lambda_{i, k}$}
		\ENDIF
	\ENDFOR
\ENDIF
\end{algorithmic}
\end{algorithm}

\begin{algorithm}
\caption{: FilterAll
  \label{algo:filterJumps2}}
\begin{algorithmic}
\IF {$j \bmod \nu = 0$}
  \STATE{$\left \{F_{v, w}\right\} \g $ max flow between all pairs of nodes $(v,w)$, computed via Gusfield's algorithm}
  \FOR{ all $i = (v,w) \in \mathscr{E}$}
	  \IF{$F_{v, w} \geq d_{\mathrm{net}}$}
		  \FOR{$k = 1$ \TO $b_i$}
			  \IF{$S_{i, k} = 1$}
				  \STATE{$S_{i, k} \g 0$}
				  \STATE{$\Lambda_{j+1} \g \Lambda_{j+1} - \tilde\lambda_{i, k}$}
			  \ENDIF
		  \ENDFOR
	  \ENDIF
  \ENDFOR
\ENDIF
\end{algorithmic}
\end{algorithm}

\subsection{Bounded relative error for PMC}
\label{sec:bre}

For a single run, the crude MC estimator $I = \bb I[\Psi(\v X(1)) < d_\mathrm{net}]$ of $u$,
which is a Bernoulli random variable with mean $u$,
has variance $u(1-u)$, so its relative error (RE) is 
$\RE[I] = \sqrt{u(1-u)}/u \approx u^{-1/2} \to \infty$ when $u\to 0$.
With $n$ runs, the variance is divided by $n$ and the RE by $\sqrt{n}$.
When $u$ is very small, we may need an excessively large $n$ to obtain a sufficiently small RE.
With PMC, the RE is sometimes much better behaved than with MC.
In this section, we obtain conditions under which the PMC estimator $W$ 
has bounded relative error (BRE), i.e., $\RE[W]$ remains bounded when $u\to 0$. 
The proofs have some similarity with those in \cite{vBOT16m}.

Suppose the probabilities $r_{i,k} = r_{i,k}(\epsilon)$ in our model depend on some parameter $\epsilon$ in a way 
that $u = u(\epsilon) \to 0$ when $\epsilon\to 0$.
In what follows, the quantities in the model are assumed implicitly to depend on $\epsilon$.
A non-negative quantity that may depend on $\epsilon$ is $\c O(1)$ if it remains bounded when $\epsilon\to 0$.
It is $\Theta(1)$ if it is bounded and also bounded away from 0, when $\epsilon\to 0$.

In our setting, the vector $\v Y$ and the permutation $\pi$ have finite length $\kappa$ 
and $C$ is bounded by $\kappa$.  The number of possible permutations is therefore finite.
Let $p(\pi)$ be the probability of permutation~$\pi$. 

\begin{proposition}
\label{co:bre-pmc-gen}
(i) If $p(\pi) = \Theta(1)$ for all $\pi$, then the PMC estimator has BRE.\\
(ii) This holds in particular if $\lambda_{i,k_i}/\lambda_{j,k_j} = \Theta(1)$ for all $i, j ,k_i, k_j$
(we then say that the rates are balanced).
\end {proposition}

\begin{proof}
(i) Note that $u \ge \bb P[T_C > 1 \mid \pi] p(\pi) = W(\pi) p(\pi)$ for any $\pi$.
If $p(\pi) = \Theta(1)$, then $W(\pi)/u = \c O(1)$ and $\max_{\pi} W(\pi)/u = \c O(1)$.
Therefore $\bb E[W^2/u^2] = \c O(1)$, which implies BRE.

(ii) Note that $p(\pi) = \prod_{j=1}^{\kappa} \lambda_{\pi(j)}/\Lambda_{\pi(j)}$.
Under the given assumption, $\lambda_{\pi(j)}/\Lambda_{\pi(j)} = \Theta(1)$ for all $j$, 
which implies that $p(\pi) = \Theta(1)$ for all $\pi$. 
\end{proof}

As a concrete illustration of an asymptotic regime in which the $r_{i,k}$ depend on $\epsilon$,
we define a regime similar to one that has been widely used for highly reliable Markovian systems 
\cite{vNAK96a,vRUB09b,vSHA94a}.
Suppose that link $i$ is  operating at capacity $c_{i,k}$.  
with probability
$$
 r_{i,k}= a_{i,k} \epsilon^{d_{i,k}}
$$
for some constants $a_{i,k}> 0$ and $d_{i,k}>0$ independent of $\epsilon$, for all $k\in \{0,\dots, b_i-1\}$
(that is, not at full capacity).  
This implies that $r_{i,b_i} = 1-\sum_{j=0}^{b_i-1} r_{i,k} = \Theta(1)$ for all $i$.
That is, the event that any link is not at full capacity is a rare event.
This implies that failure to meet the demand is a rare event, and therefore $\RE[I] \to \infty$ when $\epsilon\to 0$. 
More specifically, any state vector $\v x =(c_{1,k_1},\dots ,c_{m,k_m})$ for which $\Psi(\v x) < \dnet$ 
has probability $\bb P(\v X=\v x) = \prod_{i=1}^m r_{i,k_i} = \epsilon^{d(\v x)}(1 + o(1))$ for some $d(\v x) > 0$.
Let $d_{\min} = \min\{d(\v x) : \Psi(\v x) < \dnet\}$.  
Then 
$$
 u = \sum_{\{\v x: \Psi(\v x) < \dnet\}} \bb P(\v X=\v x) = \Theta(\epsilon^{d_{\min}}).
$$
On the other hand, we have

\begin{proposition}
In the setting just defined, the PMC estimator has BRE. 
\end{proposition}

\begin{proof}
Recall that for $1\leq i\leq m$, $\lambda_{i,b_i} = -\ln(\sum_{k=0}^{b_i-1} r_{i,k}) = \Theta(\ln(\epsilon))$.
Moreover, for all $k < b_i$, $\lambda_{i,k} = \ln(\sum_{\ell=0}^{k} r_{i,\ell}) - \ln(\sum_{\ell=0}^{k-1} r_{i,\ell}) = \Theta(\ln(\epsilon))$.
The conditions of Proposition~\ref{co:bre-pmc-gen} (ii) are then verified, hence the result.
\end{proof}

The results of this section apply  to the improved PMC variants as well;
the proofs are easily adapted.

\section{A generalized splitting algorithm}
\label{sec:gs}

Botev et al. \cite{vBOT14a} have explained how to adapt the GS algorithm 
proposed and studied in \cite{vBOT10a,vBOT13a} to the stochastic flow problem considered here,
but for the situation where the capacities have a continuous distribution.
The aim of the algorithm is to obtain a sample of realizations of $\v X$ which is approximately
a sample from the distribution of $\v X$ conditional on $\Psi(\v X) < \dnet$.
The estimator is then given by the realized sample size (which is random) 
divided by its largest possible value.
The algorithm uses intermediate demand levels $\dnet = d_\tau < \cdots < d_1 < d_0$,
where $d_0$ is the maximal possible flow, achieved when each link $i$ is at its maximal capacity
$c_{i,b_i}$.
These levels and their number $\tau$ are {fixed a priori and} chosen so that 
$\bb P[\Psi(\v X) < d_t \mid \Psi(\v X) < d_{t-1}] \approx 1/s$ for $t=1,\dots,\tau-1$,
and at most $1/s$ for $t=\tau$, where $s$ is a small integer also fixed 
(usually and in all our experiments in this paper, $s=2$). 
The levels are estimated by pilot runs, as explained in  \cite{vBOT10a,vBOT13a}.
The algorithm starts by sampling $\v X$ from its original distribution.
If $\Psi(\v X) < d_1$, it resamples each coordinate of $\v X$ conditional on $\Psi(\v X) < d_1$,
via Gibbs sampling, repeats this $s$ times, 
and keeps the states $\v X$ for which  $\Psi(\v X) < d_2$ 
(their number is in $\{0,\dots,s\}$).
At each level $t=3,\dots,\tau$, this type of resampling is applied to each state that has been retained
at the previous step (for which  $\Psi(\v X) < d_{t-1}$), by resampling that state twice from its 
distribution conditional on $\Psi(\v X) < d_{t-1}$, and retaining the states for which 
$\Psi(\v X) < d_{t}$.  At the last level, we count the number $N$ of chains for which 
$\Psi(\v X) < d_{\tau}$, and return $W = N/s^{\tau-1}$ as an estimator of $u$.
This is repeated $n$ times independently, to produce $n$ independent realizations of $W$,
say $W_1,\dots, W_n$, whose average $\bar W_n$ is an unbiased estimator of $u$.
This estimator \emph{does not} have BRE, because the RE increases with the number
of levels; the RE is typically (roughly) proportional to $-\log u$ 
(see \cite{vBOT12a} for a proof of this result in an idealized setting).
It can also handle large networks.

In general, this GS algorithm is not directly applicable when the capacities
have discrete distributions, because then $\Psi(\v X)$ also has a discrete distribution
and it may happen that this distribution is too coarse
(e.g., all the probability mass is on just a few possible values).
Then it may be impossible to select levels $d_t$ for which
$\bb P[\Psi(\v X) < d_t \mid \Psi(\v X) < d_{t-1}] \approx 1/s$ for $t=1,\dots,\tau-1$.

It is nevertheless possible to apply GS in that case by constructing the vector $\v Y$
as for PMC in the previous section, and resampling this vector instead of $\v X$.
Recall that $\v Y$ is a vector of $\kappa$ independent exponential random variables.
The GS algorithm will operate similarly as the one described above, except that now
the levels $0 = \gamma_0 < \gamma_1 < \cdots < \gamma_\tau = 1$ are on $T_C$,
and the resampling at each step $t$ is for $\v Y$ and is conditional on $T_C > \gamma_{t-1}$.
This is valid because $\{\Psi(\v X(\gamma_t)) < d_\mathrm{net}\} = \{T_C > \gamma_t\}$.
%
The corresponding GS procedure operates in the same way as the GS algorithm with anti-shocks in \cite{vBOT16m}.
We first generate $\v Y$ from its original distribution.
Then at each level $t$, we take each state (realization or modification of $\v Y$) that has been retained   
at the previous step (for which  $T_C > \gamma_{t-1}$), we resample all its coordinates $s$ times 
(i.e., for $s$ Gibbs sampling steps, where each step starts from the result of the previous step) 
from its distribution conditional on $T_C > \gamma_{t-1}$, to obtain two new states, 
and we retain the states for which $T_C > \gamma_{t}$.  
At the last level, we count the number $N$ of chains for which 
$T_C > \gamma_{\tau} = 1$, and return $W = N/s^{\tau-1}$ as an estimator of $u$.
The resampling of $\v Y$ conditional on $T_C > \gamma_{t-1}$ via Gibbs sampling
can be done in a similar way as in \cite{vBOT13a,vBOT16m}.
We first select a permutation $\pi$ of the $\kappa$ coordinates of the vector $\v Y$.
Then for $j = 1,\dots,\kappa$,  we resample $Y_{\pi(j)}$ as follows:
If $\pi(j) = (i,k)$, the current capacity $X_i(\gamma_{t-1})$ of link $i$ is less than $c_{i,k}$
(or equivalently $\min(Y_{i,k},\dots,Y_{i,b_i}) > \gamma_{t-1}$), 
and by changing the current capacity of link $i$ to $c_{i,k}$
(or equivalently changing $Y_{i,k}$ to 0) we would have $T_C < \gamma_{t-1}$ 
(the maximum flow would meet the demand), then we resample $Y_{i,k}$ from its exponential density 
truncated to $(\gamma_{t-1}, \infty)$. Otherwise we resample  $Y_{i,k}$ from its original exponential density. 
To sample from the truncated density, it suffices to generate $Y_{i,k}$ from the original density and 
add $\gamma_{t-1}$.


\section{Numerical Examples}
\label{sec:examples}

In this section, we provide some numerical examples that compare 
the PMC and GS algorithms, and show how they behave when $u\to 0$.
In these examples, we parameterize the models by $\epsilon$ in a way that 
$u = u(\epsilon) \to 0$ when  $\epsilon\to 0$, exactly as in Section~\ref{sec:bre}, 
{in the asymptotic regime when the probability that links are not operating at 
full capacity is getting close to zero}. 
For all variants of PMC, we used formula (2) in \cite{vBOT16m} with high precision arithmetic 
to compute (\ref{eq:sum-cdf}). 

\subsection{Experimental setting}

We used the same experimental protocol as in \cite{vBOT16m}, comparing four methods. Method PMC refers to Algorithm~\ref{algo:pmc-anti} without any filtering step. PMC-Single and PMC-All refer to PMC combined with the filtering as in Algorithms~\ref{algo:filterJumps1} and \ref{algo:filterJumps2} respectively. Method GS refers to generalized splitting, implemented as described in Section~\ref{sec:gs}. The splitting levels were determined via the adaptive Algorithm 3 of \cite{vBOT13a}, with $n_0 = 500$ and $s=2$. The levels were estimated using a single run of the adapative algorithm, and these same levels were used for every independent replication of the GS algorithm. 

For each example and method,  
we report the unreliability estimate $\bar W_n$,
its empirical \emph{relative error} $\RE[\bar W_n] = S_n/ (\sqrt{n} \bar W_n)$
where $S_n^2$ is the empirical variance,
and the \emph{work-normalized relative variance} (WNRV) of $\bar W_n$,
defined as $\WNRV[\bar W_{\!n}] = T \times \RE^2[\bar W_{\!n}]$,
where $T$ is the total CPU time (in seconds) for the $n$ runs of the algorithm.
One must keep in mind that $T$ and the WNRV depend on the software and hardware used for the computations.
The experiments were run on Intel Xeon E5-2680 CPUs, on a linux cluster. 
The sample size for every algorithm was $n = 5 \times 10^4$. 

For each example we use the following model. 
Each link $i$ has the capacity levels $\{0,1,\dots, b_i\}$, 
i.e., $c_{i,k} = k$ for $k=0,\dots,b_i$. We take $r_{i,k} = \bb P(X_i = c_{i,k}) = \rho^{b_i-k-1} \epsilon$ for $k < b_i$ and 
$r_{i,b_i} = 1 - \sum_{k=0}^{b_i-1} \rho^{b_i-k-1} \epsilon$,
where $\rho$, $\epsilon$ and $\left\{ b_i\right\}$ are model parameters.


\subsection{A $4 \times 4$ lattice graph}

Our first example uses the $4 \times 4$ lattice graph, which has $16$ nodes and $24$ links.
The flow has to be sent from one corner to the opposite corner.
We take $b_i = 8$, $\rho = 0.6$ and $\dnet = 10$, and let $\epsilon$ range from 
$10^{-4}$ to $10^{-13}$. 

Table~\ref{tab:grid-results_4} reports the values of 
$\bar W_n$, $\RE[\bar W_n]$ and $\WNRV[\bar W_{\!n}]$, for methods GS and PMC, for some values of $\epsilon$. 
Figure \ref{fig:wnrv_grid_4} shows plots of $\RE[\bar W_n]$ and $\WNRV[\bar W_{\!n}]$ for all four methods. 
We see that for this small example, PMC-All always has the smallest RE, followed by PMC-Single.
These REs increase very slowly when $\epsilon$ decreases, for the values we have tried.
They should eventually stabilize when $\epsilon\to 0$.
In terms of work-normalized relative variance, i.e., when taking the computing time into account,
GS is the most effective method when $\epsilon$ is not too small, but when $\epsilon$ gets smaller, 
GS requires a larger simulation effort while the effort required by PMC variants remains approximately stable, 
so these PMC methods eventually catch up, in agreement with our asymptotic results.
Among them, PMC-Single has the smallest WNRV.

\begin{table}[ht]
\centering
\begin{tabular}{l|rrrrr}
  \hline
 & $\varepsilon = 10^{-4}$ & $\varepsilon = 10^{-5}$ & $\varepsilon = 10^{-6}$ & $\varepsilon = 10^{-7}$ & $\varepsilon = 10^{-8}$ \\
  \hline
$\bar W_n$ for PMC & $2.99 \times 10^{-5}$ & $2.98 \times 10^{-6}$ & $2.99 \times 10^{-7}$ & $2.99 \times 10^{-8}$ & $2.99 \times 10^{-9}$ \\
  RE[$\bar W_{\!n}$] for PMC & $3.16 \times 10^{-2}$ & $3.34 \times 10^{-2}$ & $3.41 \times 10^{-2}$ & $3.69 \times 10^{-2}$ & $3.74 \times 10^{-2}$ \\
  WNRV[$\bar W_{\!n}$] for PMC & $3.17 \times 10^{-2}$ & $3.20 \times 10^{-2}$ & $3.17 \times 10^{-2}$ & $3.62 \times 10^{-2}$ & $3.54 \times 10^{-2}$ \\
  $\bar W_n$ for GS & $2.98 \times 10^{-5}$ & $2.99 \times 10^{-6}$ & $2.99 \times 10^{-7}$ & $2.98 \times 10^{-8}$ & $2.99 \times 10^{-9}$ \\
  RE[$\bar W_{\!n}$] for GS & $3.43 \times 10^{-2}$ & $3.32 \times 10^{-2}$ & $3.15 \times 10^{-2}$ & $3.30 \times 10^{-2}$ & $4.33 \times 10^{-2}$ \\
  WNRV[$\bar W_{\!n}$] for GS & $1.07 \times 10^{-2}$ & $1.43 \times 10^{-2}$ & $1.68 \times 10^{-2}$ & $2.35 \times 10^{-2}$ & $2.86 \times 10^{-2}$ \\
   \hline
\end{tabular}
\caption{Estimation of $u$, RE, and WNRV for some values of $\epsilon$, for the $4 \times 4$ lattice example}
\label{tab:grid-results_4}
\end{table}

\begin{figure}
\centering
\includegraphics[scale=0.45]{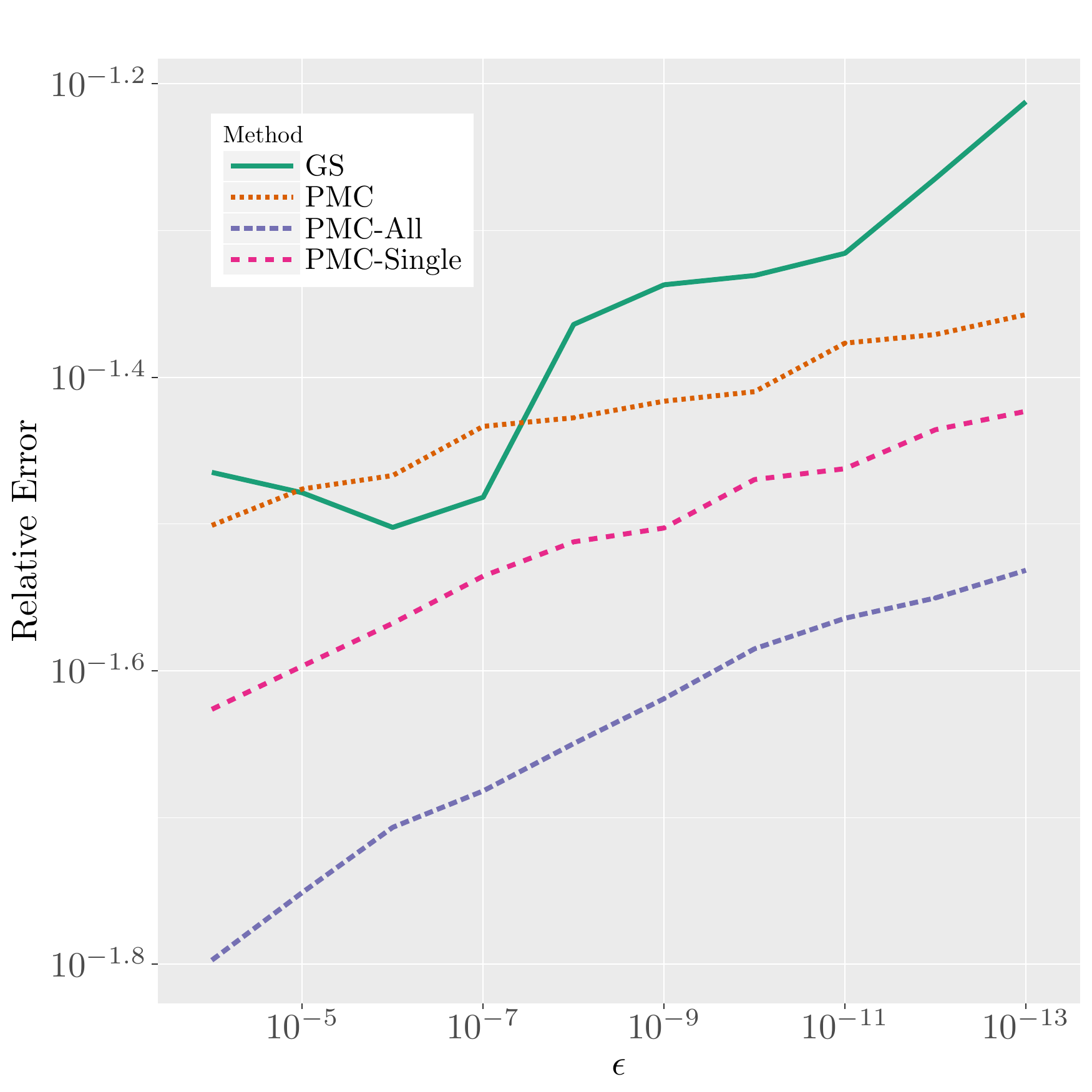} 
\hskip 10pt
\includegraphics[scale=0.45]{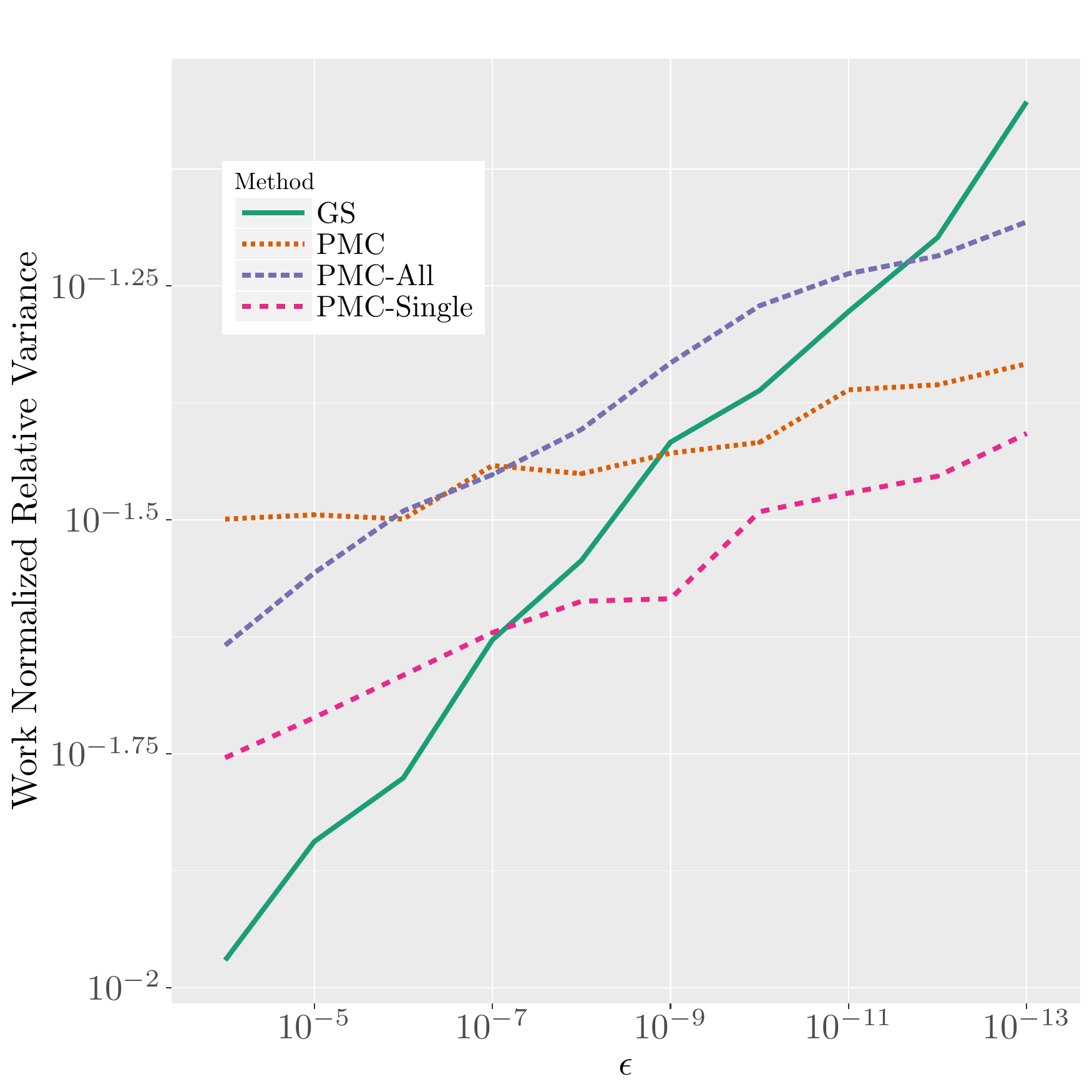}
\caption{RE (left) and WNRV (right) for four methods, for the $4 \times 4$ lattice graph. 
\label{fig:wnrv_grid_4}}
\end{figure}

\subsection{$6 \times 6$ lattice graph}

\begin{figure}
\centering
\includegraphics[scale=0.45]{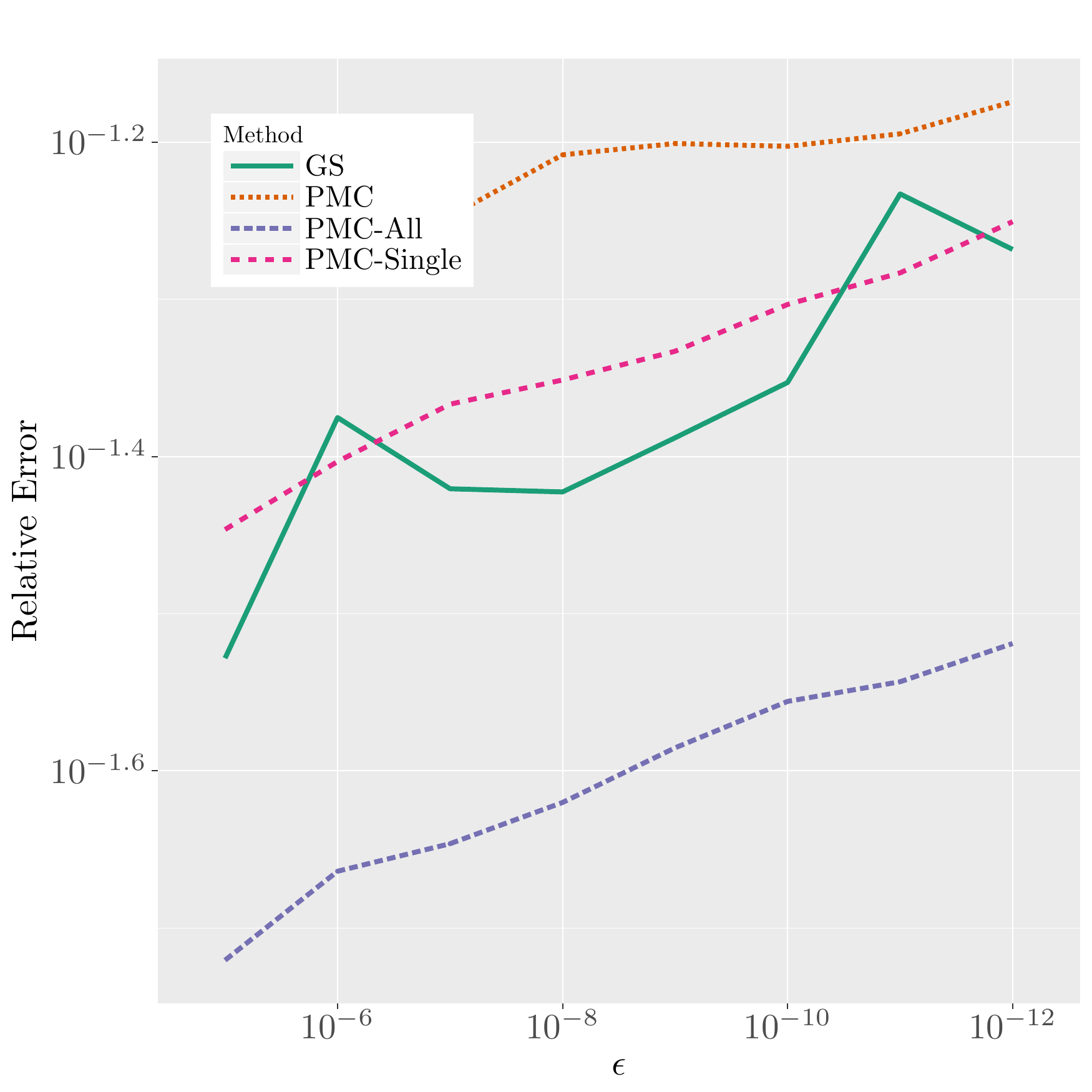} 
\hskip 10pt
\includegraphics[scale=0.45]{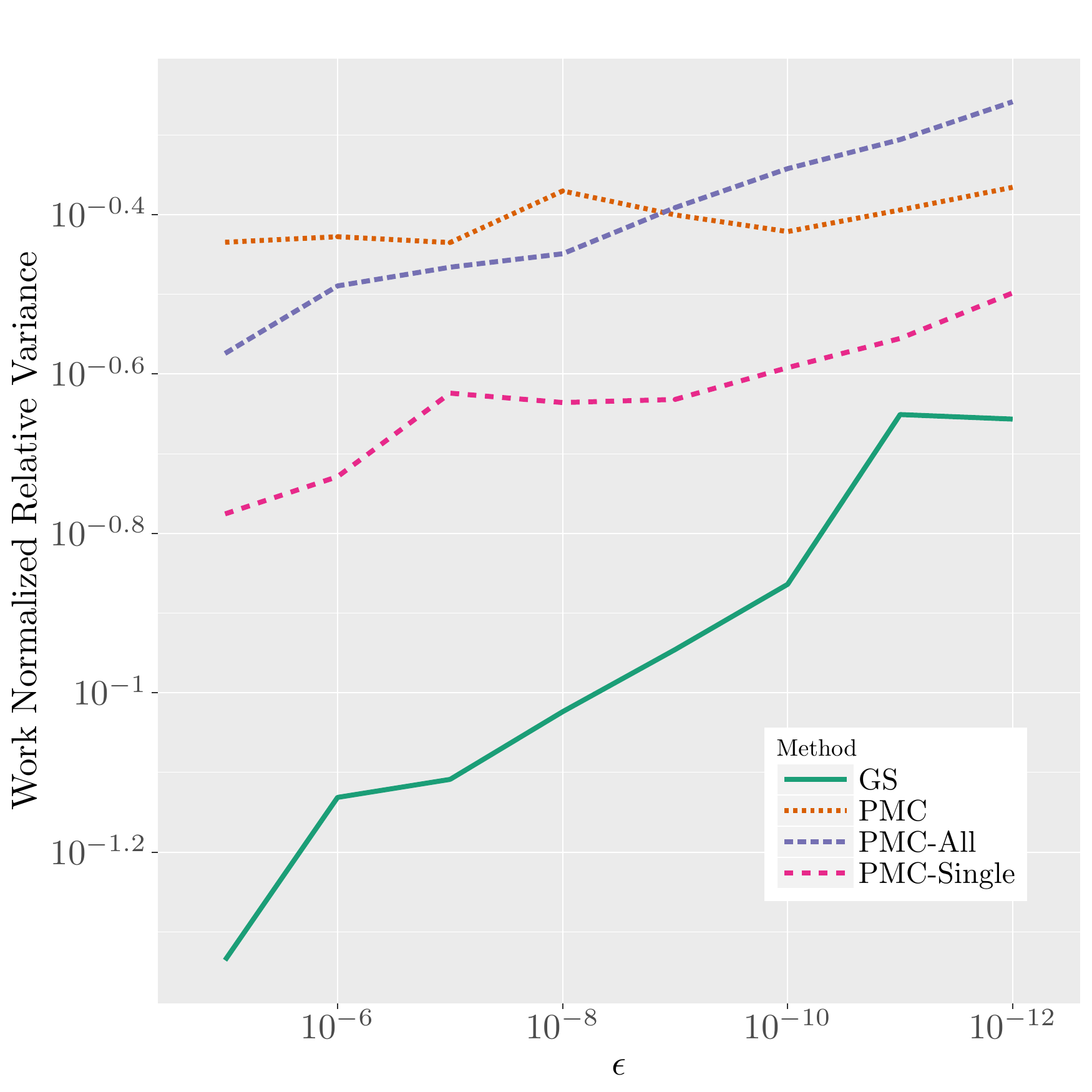}
\caption{RE (left) and WNRV (right) for four methods, for the $6 \times 6$ lattice graph. 
 \label{fig:wnrv_grid_6}}
\end{figure}

Figure \ref{fig:wnrv_grid_6} shows plots of $\RE[\bar W_n]$ and $\WNRV[\bar W_{\!n}]$ for a $6\times 6$ 
lattice graph, with 36 nodes and 60 links.  
Again, the flow has to be sent from one corner to the opposite corner.
Here, for all values of $\epsilon$ considered, PMC-All has the smallest RE while GS wins in terms of WNRV.

\subsection{A dodecahedron network}

\begin{figure}
\begin{center}
\includegraphics[scale=0.4]{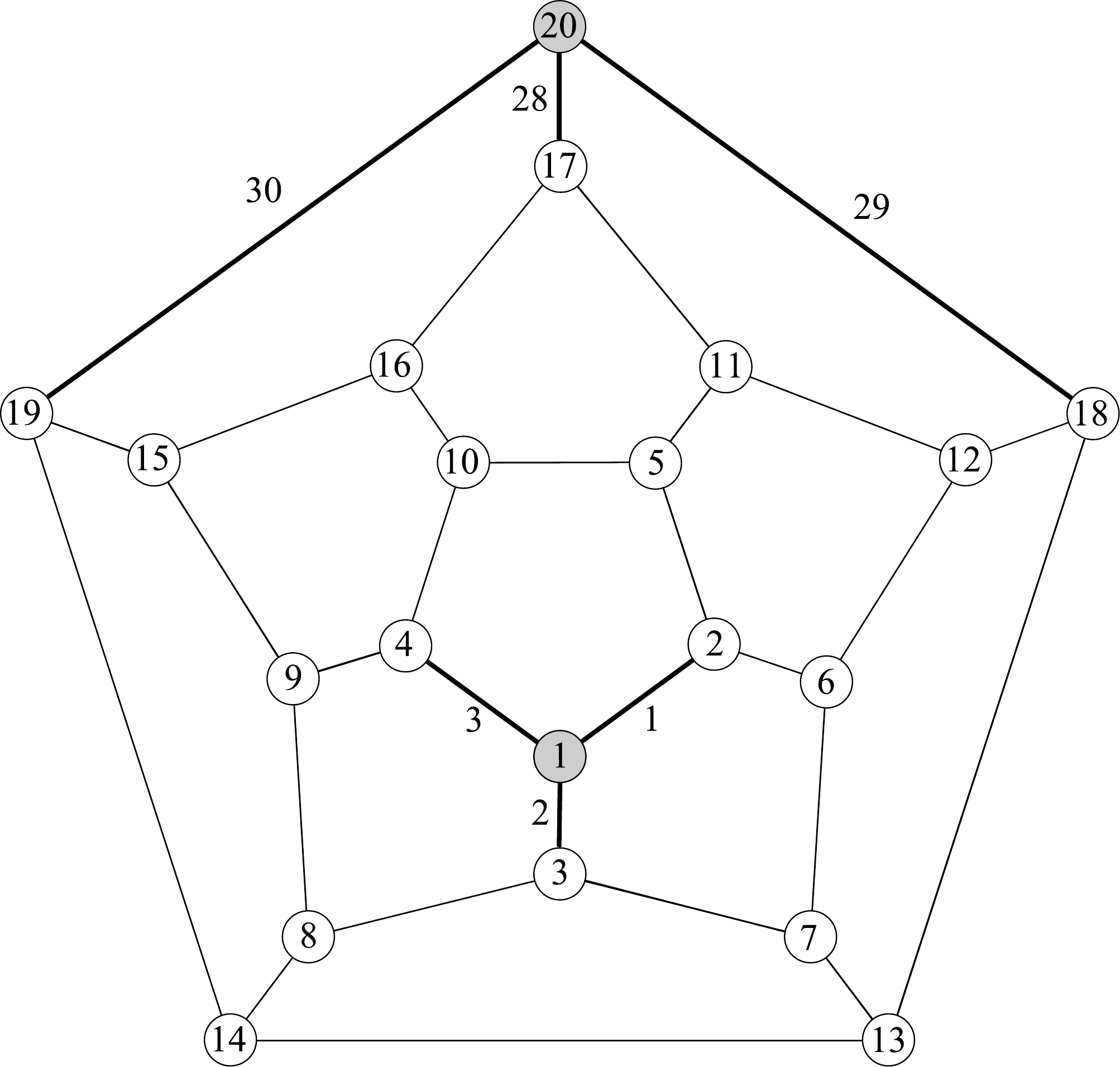}
\end{center}
\caption{A dodecahedron graph with 20 nodes and 30 links (figure taken from \cite{vBOT13a}).}
\label{fig:dode}
\end{figure}

In this example we use the well-known dodecahedron network (Figure~\ref{fig:dode}),
with 20 nodes and 30 links, often used as a standard benchmark in network reliability 
estimation \cite{vBOT13a,vCAN95a,vCAN09a,vCAN09b,vTUF14a}. 
Here we took $\rho = 0.7$, $b_i = 4$ and $d_{\mathrm{net}} = 5$. 
Note that when $\epsilon$ is very small, most of the failures will occur because there is 
not enough capacity in the three links connected to node 1 (links 1, 2, 3), 
or not enough capacity in the three links connected to node 20 (links 28, 29, 30).
These are the two bottleneck cuts.

Table~\ref{tab:dode-results} reports the values of 
$\bar W_n$, $\RE[\bar W_n]$, and $\WNRV[\bar W_{\!n}]$, for the GS and PMC methods, 
for different values of $\epsilon$. 
We see that the estimates $\bar W_n$ agree very well across the two methods.
Figure \ref{fig:wnrv_dodec} shows $\RE[\bar W_n]$ and $\WNRV[\bar W_{\!n}]$ as functions of $\epsilon$,
for all four methods.
We see that PMC-All has by far the smallest RE for all $\epsilon$, and it also wins in terms 
of WNRV, except for $\epsilon > 10^{-5}$ where GS wins.
The latter case is approximately when $u \ge 7\times 10^{-11}$, which is already pretty small. 
When $\epsilon$ decreases, the WNRV increases for GS in part because the RE increases,
but also because the computing time increases.
The figure shows what happens when $\epsilon$ gets very small.

\begin{table}[ht]
\centering
\begin{tabular}{l|rrrrr}
  \hline
 & $\varepsilon = 10^{-4}$ & $\varepsilon = 10^{-5}$ & $\varepsilon = 10^{-6}$ & $\varepsilon = 10^{-7}$ & $\varepsilon = 10^{-8}$ \\
  \hline
$\bar W_n$ for PMC & $7.08 \times 10^{-9}$ & $7.08 \times 10^{-11}$ & $7.06 \times 10^{-13}$ & $7.06 \times 10^{-15}$ & $7.05 \times 10^{-17}$ \\
  RE[$\bar W_{\!n}$] for PMC & $8.63 \times 10^{-2}$ & $7.21 \times 10^{-2}$ & $6.68 \times 10^{-2}$ & $5.97 \times 10^{-2}$ & $5.86 \times 10^{-2}$ \\
  WNRV[$\bar W_{\!n}$] for PMC & $1.85 \times 10^{-1}$ & $1.37 \times 10^{-1}$ & $1.14 \times 10^{-1}$ & $8.95 \times 10^{-2}$ & $8.57 \times 10^{-2}$ \\
  $\bar W_n$ for GS & $7.07 \times 10^{-9}$ & $7.06 \times 10^{-11}$ & $7.07 \times 10^{-13}$ & $7.07 \times 10^{-15}$ & $7.05 \times 10^{-17}$ \\
  RE[$\bar W_{\!n}$] for GS & $3.95 \times 10^{-2}$ & $4.30 \times 10^{-2}$ & $4.58 \times 10^{-2}$ & $5.17 \times 10^{-2}$ & $4.97 \times 10^{-2}$ \\
  WNRV[$\bar W_{\!n}$] for GS & $3.13 \times 10^{-2}$ & $4.52 \times 10^{-2}$ & $6.00 \times 10^{-2}$ & $8.81 \times 10^{-2}$ & $1.06 \times 10^{-1}$ \\
   \hline
\end{tabular}
\caption{Estimation of $u$, RE, and WNRV for some values of $\epsilon$, for the dodecahedron example}
\label{tab:dode-results}
\end{table}

\begin{figure}
\centering
\includegraphics[scale=0.45]{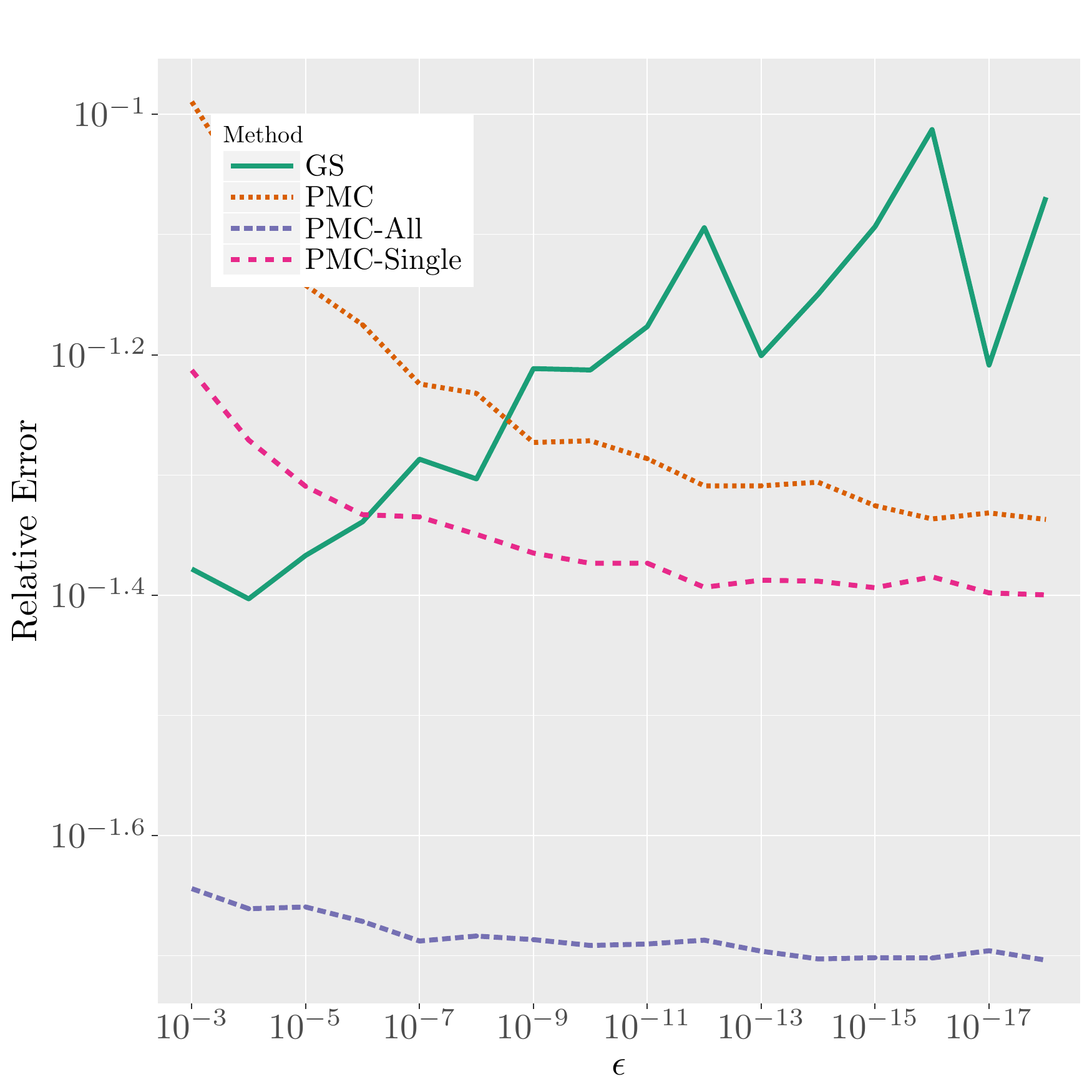} 
\hskip 10pt
\includegraphics[scale=0.45]{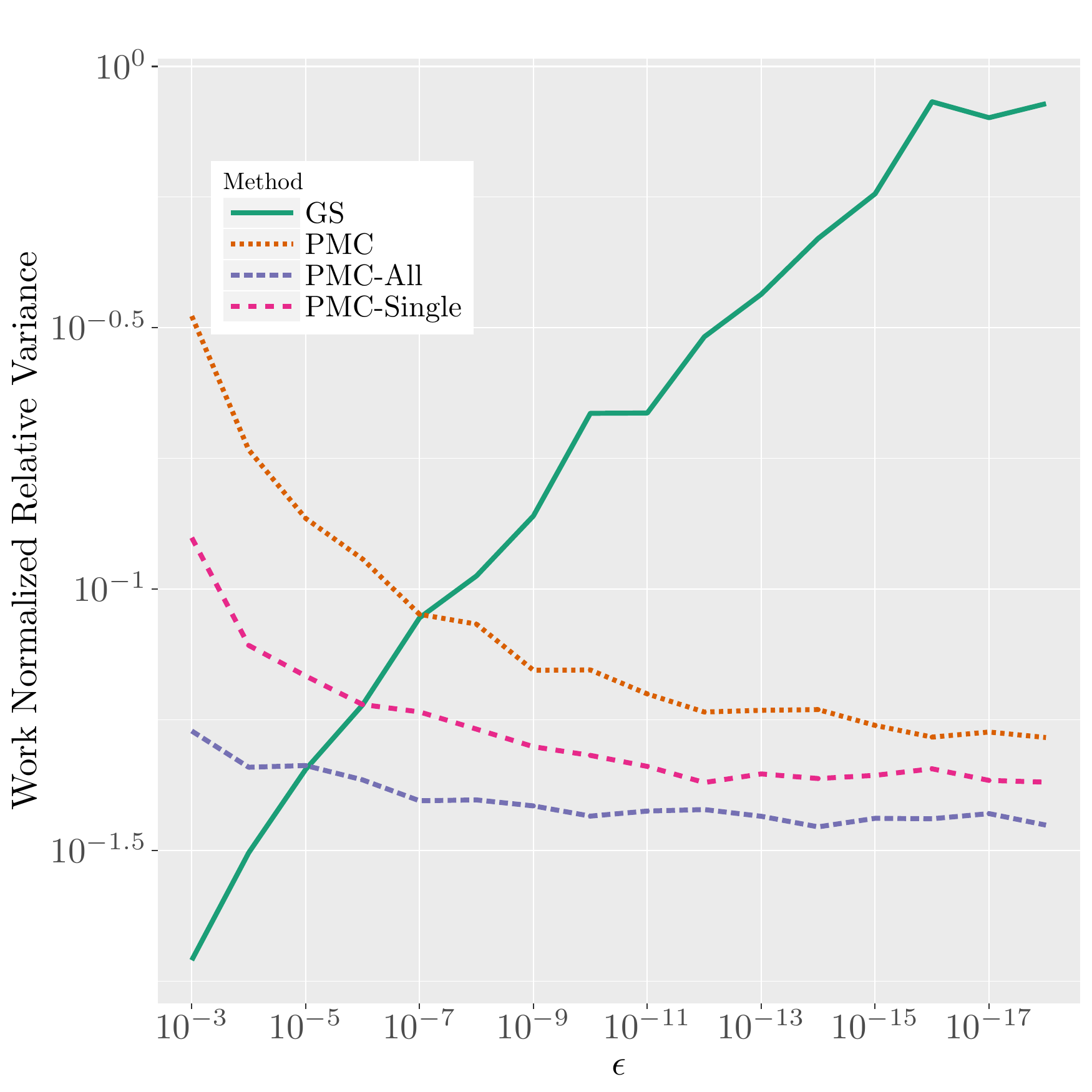}
\caption{RE (left) and WNRV (right) for four methods, for the dodecahedron example}
\label{fig:wnrv_dodec}
\end{figure}

\section*{Acknowledgments}

This work has been supported by the Australian Research Council under DE140100993 Grant to Z. I. Botev,
an NSERC-Canada Discovery Grant,
a Canada Research Chair, and an Inria International Chair to P. L'Ecuyer.
P. L'Ecuyer acknowledges the support of the Faculty of Science Visiting Researcher Award at UNSW.
We are grateful to Rohan Shah, who performed the numerical experiments.

\bibliographystyle{plain}

\end{document}